\documentclass[11pt]{article}

\usepackage{mathptmx}
\usepackage{amssymb,amsmath,amsthm,amstext}

\newcommand{\bigO}[1]{\ensuremath{\operatorname{O}\bigl(#1\bigr)}}
\DeclareMathOperator{\C}{\mathit{C}\,}
\DeclareMathOperator{\CT}{\mathit{CT}\,}

\usepackage{color}
\newcommand{\ver}[1]{#1}

\newtheorem{theorem}{Theorem}[section]
\newtheorem{lemma}[theorem]{Lemma}

\theoremstyle{remark}
\newtheorem{remark}{Remark}[section]

\newtheorem{definition}[theorem]{Definition}

\newcommand{\zo}{\{0,1\}}
\newcommand{\ie}{i.e.}
\newcommand{\poly}{\text{poly}}

\newcommand{\eps}{\varepsilon}

\newcommand{\RIGHT}{\rm RIGHT}
\newcommand{\LEFT}{\rm LEFT}
\usepackage{epsfig}
\usepackage{amsmath}

\usepackage{amsfonts}

\begin{document}

\title{Short lists with short programs in short time \footnote{An extended abstract of this work has been presented at the $28$th IEEE Conference in Computational Complexity, Stanford, CA, June 5-7, 2013.}}
\author{{Bruno Bauwens\/}
\thanks{National Research University Higher School of Economics.
     This work was partially supported by the NAFIT ANR-08-EMER-008-01 project.}
\and{{Anton Makhlin}\/}
\thanks{Moscow State University. The work
was in part supported by the RFBR grant 12-01-00864 and the ANR grant
ProjetANR-08-EMER-008; email: amakhlin@bk.ru}
\and{{Nikolay Vereshchagin}\/}
\thanks{National Research University Higher School of Economics. The work
was in part supported by the RFBR grant 16-01-00362.
E-mail: ver@mccme.ru, 
WWW home page: 
http://lpcs.math.msu.su/$\widetilde{\hskip 1ex}$ver. }
\and
{ {Marius Zimand\/}
\thanks{  Department of Computer and Information Sciences, Towson University,
Baltimore, MD.; email: mzimand@towson.edu; http://triton.towson.edu/\~{ }mzimand.
The author was supported in part by NSF grant CCF 1016158.}
}
}
\date{}

\maketitle

\begin{abstract} 
  Given a machine $U$, a $c$-short program for $x$ is a string $p$ such that
  $U(p)=x$ and the length of $p$ is bounded by $c$ + (the length of a shortest
  program for $x$).  We show that for any standard Turing machine, it is
  possible to compute  in polynomial time  on input $x$ a  list of polynomial
  size guaranteed to contain a $\bigO{\log |x|}$-short program  for $x$.  We
  also show that there exists a computable function that maps every $x$ to a
  list of size $|x|^2$ containing a $\bigO{1}$-short program for $x$.
  This is essentially optimal because we prove that for each such 
  function there is a $c$ and infinitely many $x$ for which the list has size
  at least $c|x|^2$.  Finally we show that for some standard machines,
  computable functions generating lists with $0$-short programs, must have
  infinitely often list sizes proportional to $2^{|x|}$. 
\end{abstract}


\section{Introduction}
The Kolmogorov complexity of a string $x$  is the length of a shortest program computing it. Determining the Kolmogorov complexity of a string  is a  canonical example of  a function that is not computable. 
Closely related, and non-computable as well,  is the problem of actually producing a shortest program for $x$. 
It is natural to ask if tasks marred by such an impossibility barrier can  be   effectively solved at least in some approximate sense.
This issue has been investigated for Kolmogorov complexity in various ways. First of all,  it is well-known that the Kolmogorov complexity can be effectively approximated from above. A different type of approximation is given by what is typically called \emph{list computability} in algorithms and complexity theory and \emph{traceability} (\emph{enumeration}) 
in  computability theory. For this type of approximation, one would like to compute a list of ``suspects'' for the result of the function with the guarantee that the actual result is in the list. Of course, the shorter the list is, the better is the approximation.

The list approximability of the Kolmogorov complexity, $C(x)$, has been studied by Beigel et al.~\cite{bei-bur:j:enumerations}.  They observe that $C(x)$ can be approximated by a list of size $(n-a)$ for every constant $a$, where $n = |x|$. On the other hand, they show that, for every universal machine $U$, there is a constant $c$ such that for infinitely many strings $x$ (in fact for at least one $x$ at each sufficiently large length $n$), any computable list containing $C_U(x)$ must have size larger than $n/c$.

In this paper we study  list approximability for the problem of producing short programs. In order to describe our results, we need several formal definitions.

A machine $U$ is {\em optimal} if $\C_U(x|y) \le C_V(x|y) + \bigO{1}$ for all machines $V$  (where the
constant $\bigO{1}$ may depend on $V$).   
An optimal machine $U$ is \emph{standard}\footnote{
This notion was introduced by Schnorr~\cite{schnorr}, and he called such machines \emph{optimal
G\"odel numberings} (of the family of all computable functions from strings to strings). 
We use a different term 
to distinguish between optimal functions in Kolmogorov's sense and Schnorr's sense},  
if for every machine $V$ there is a 
total computable function $t$ such that for all $p,y$: $|t(p)| = |p| +\bigO{1}$ and $U(t(p),y) = V(p,y)$ if $V(p,y)$ is defined. 
For results that hold in polynomial time, we additionally assume these functions $t$ can be computed in
time polynomial in $|p|$.
Let $U(p)$ stand for $U(p,\text{the empty string})$ and  $C_U(x)$ for $C_U(x|\text{the empty string})$.
A {\em $c$-short program} for $x$ with respect to  $U$ is a string $p$ that satisfies
$U(p) = x$ and $|p| \le \C_U(x) +c$. 

Given an optimal machine $U$,
a \emph{list-approximator} for $c$-short programs is a function $f$ that on every input $x$ outputs a finite list of strings such that at least one of the elements in the list is a $c$-short  program for $x$ on $U$.  
Let $|f(x)|$ denote the number of elements in the list  $f(x)$. Obviously, for every optimal $U$, there is a (trivial) computable list-approximator $f$ such that $|f(x)| \leq 2^{|x|+\bigO{1}}$. 

The question we study is how small can $|f(x)|$ be for computable list-approximators $f$ for $c$-short programs, where $c$ 
is a constant or $\bigO{\log|x|}$.
At first glance it seems that in both cases  $|f(x)|$ must be exponential in $|x|$.
Surprisingly, this is not the case. We prove that there is a computable approximator
with list of size $|x|^2$ for $c$-short programs for some constant $c$ depending on the choice of the standard machine $U$.
And we show that this bound is tight.
We show also that there is a \emph{polynomial time computable}
approximator with list of size $\poly(|x|)$ for  $c$-short programs for $c=\bigO{\log|x|}$.

We start with our main upper bound. 
We show  that for every standard machine, there exists a list-approximator for $\bigO{1}$-short programs, with lists of \emph{quadratic} size.  

\begin{theorem} \label{t:complist}
For every standard machine $U$ there exist  a constant $c$ and a computable function $f$ that for any $x$ produces a
list with  $|x|^2$ many elements containing a program $p$ for $x$ of length $|p| \le C_U(x) + c$. 
\end{theorem}

\noindent
The constant $c$  depends on $U$. The next 
result shows that this dependence is unavoidable.

\begin{theorem}\label{th:computeSetsShortest}
(1) For every $c$ there 
exists a standard\/\footnote{The construction implies the existence of such $U$ with a stronger universality property:
for every machine $V$ there exists a string $w_V$ such that $U(w_V p) = V(p)$ for all $p$.}
machine $U$ such that for every  computable $f$ that is a  list-approximator for $c$-short programs:
\[
|f(x)| \ge e 2^{|x|},
\]
for some constant $e > 0$ and infinitely many $x$.\footnote{
  For $c=0$ this was independently obtained by Frank Stephan~\cite{ste:personalcomm:shortlistMachine}.
  }

(2) On the other hand, for every $c$ 
there is a standard\footnotemark[2] machine $U$ which has a computable
approximator  for $c$-short programs with lists  of size $|x|^2$.
(As any approximator of 0-short programs is an 
approximator for $c$-short programs for every $c\ge0$ as well, only the case $c=0$ matters.) 
 \end{theorem}

Thus for any $c$ the answer to the question of whether a standard machine 
$U$ has a computable
approximator for $c$-short programs with polynomial size lists depends on the choice of $U$.

If we allow $\bigO{\log |x|}$-short  programs we can construct lists of polynomial size in \emph{polynomial time}.

\begin{theorem} \label{t:polytimelist}
  For every standard machine $U$, there exists a polynomial-time computable function $f$  
  that for any $x$ produces a list with $\poly(|x|)$ many elements containing  a program for $x$ 
  of length $C_U(x)+ \bigO{\log |x|}$.
\end{theorem}

\noindent
\ver{Teutsch}~\cite{teu:j:shortlists} has improved this result 
replacing $\bigO{\log |x|}$ by~$\bigO{1}$. A further  improvement that brings the list size to $O(|x|^{6+\epsilon})$ has been 
obtained by Zimand~\cite{zim:c:shortlistshortproof}.

\medskip
Now we move to lower bounds that hold for all standard machines. 
A linear lower bound for the list size 
of any computable approximator for 
$c$-short programs can be easily derived
from the result of Beigel et al.~\cite{bei-bur:j:enumerations} cited earlier. Indeed, if there is  
a computable list-approximator for $c$-short programs of $x$ with list of size $s(x)$,
then we can algorithmically produce a list of size $s(x)c$ containing $C(x)$. 
Thus $s(x)c=\Omega(|x|)$ and $s(x)=\Omega(|x|/c)$.
(The constant hidden in the $\Omega$-notation depends on the universal machine $U$.) 
A  linear lower bound also follows
from a result of  Bauwens~\cite{bau-shen:j:compcomp}, improving a theorem of
G\'{a}cs~\cite{gac:j:symmetry}. The result states that for all universal machines $U$, 
$\C_U(C_U(x) \mid x)$ is greater than $\log |x| - \bigO{1}$ 
for infinitely many $x$.  Thus, for any computable approximator $f(x)$ for $c$-short programs for infinitely many $x$, 
\[
\log |x| - \bigO{1} \leq \C_U(\C_U(x) \mid x) \leq \log |f(x)| + 2 \log c + \bigO{1},
\]
and therefore $|f(x)| \geq \Omega(|x|/c^2)$. 

Note  the gap between the quadratic upper bound for list size of Theorem \ref{t:complist}
and the linear lower bound described above. 
We close this gap by showing that 
the  list size bound in Theorem \ref{t:complist} is optimal:
it is not possible to compute lists of subquadratic size that contain a $\bigO{1}$-short program. 

\begin{theorem}\label{th:quadraticlistt}
For all $c>0$, for every optimal $U$, for every computable $f$ that is  a  list-approximator for $c$-short programs,
\[
|f(x)| \ge \Omega(|x|^2/c^2),
\]
for infinitely many $x$.
(The constant hidden in the $\Omega$-notation depends on the function $f$ and machine $U$.)
\end{theorem}



 \smallskip

\textbf{Technical overview.}
A $c$-short program of a string $x$ is a compressed representation of $x$;  it is also a string of (close to) minimal length that retains all the randomness in $x$. The idea that comes to mind to approach list-approximability of $c$-short programs  is to use randomness extractors. Kolmogorov complexity extraction has been studied before by Fortnow, Hitchcock, Pavan, Vinodchandran, Wang, Zimand~\cite{fhpvw:j:extractKol,hit-pav-vin:j:Kolmextraction,zim:j:extractKolm,zim:c:symkolm,zim:j:kolmindep} (see also the survey paper by Zimand~\cite{zim:j:kolmextractsurvey}), and indeed randomness extractors for a constant number of independent sources have been employed for this task. For the list-approximability of $c$-short programs, it seems natural to use seeded extractors, because by iterating over all possible seeds, one obtains a list containing the optimally compressed string.
The problem is that  we need an extractor with logarithmic seed (because we want a list of polynomial size) and zero entropy loss (because we want  the compressed string to be a program for $x$,  \ie, to contain enough information so that $x$ can be reconstructed from it) and   such extractors have not yet been shown to exist.  Perhaps surprisingly,  simpler graphs satisfying less demanding combinatorial constraints than extractors graphs, are good enough for the list approximation of $c$-short programs (and also for this type of ``list"-extraction  of  Kolmogorov complexity). Inspired by the work of Musatov, Romashchenko and Shen~\cite{mus-rom-she:j:muchnik},  we use graphs that allow on-line matching.  These are unbalanced bipartite graphs, which, in their simplest form, have $\LEFT = \zo^n, \RIGHT=\zo^{k + \mbox{small overhead}}$, and left degree = $\poly(n)$, and  which permit on-line matching up to size $K=2^k$. This means that any set  of $K$  left nodes,  each one requesting   to be matched to some adjacent right node that was not allocated earlier, can be satisfied in the on-line manner (i.e., the requests arrive one by one and  each request is satisfied before seeing the next one; in some of our proofs we will allow a small number of requests to be discarded, but this should also happen before the next request arrives).  The correspondence to our problem is roughly that left strings are the strings that we want to compress, and  for any left $x$  we seek  its compressed form among its right neighbors.  To understand this correspondence, let us consider the easier situation of producing a short list containing a $c$-short  program for $x$ in case the program knows $n = |x|$. Also, let us assume that the Kolmogorov complexity of $x$ is $k$.  We start an enumeration of strings produced by $k$-bit programs, and when a string of length $n$ is enumerated (eventually, $x$), we use the on-line matching process and find it a match among its right neighbors, \ie, we compress it to length $k$ + small overhead.  To decompress,  we start with a right node (the compressed string) and we re-play the enumeration and  the matching process and see which left node has been matched to it; we output this left node.  The compressed program for $x$  is among its right neighbors and, therefore the set of right neighbors of $x$  is the desired list of polynomial size. Now, in fact,  neither $n$ nor $k$  are known,  and therefore we actually need the bipartite graph to be infinite and the matching requests for a left node $x$ to include the desired ($k + $ overhead) length for the matching right node. 
With this modification, it turns out that the construction of a list approximator $f$ for $c$-short programs is equivalent to the construction of an infinite bipartite graph $G$ that can satisfy the on-line matching requests with overhead equal to $c + \bigO{1}$. The size of the list $f(x)$ is equal to the degree of $x$ in $G$.
Such an infinite graph is obtained by taking the union of finite graphs of the type described above.
In order for a finite  graph to allow matching, it needs to have good expansion properties. It turns out that it is enough if left subsets of   size $K/\bigO{1}$ expand to size $K$. To obtain the graph required in Theorem~\ref{t:complist}, we use the probabilistic method (actually to get quadratic left degree, we need to refine the construction sketched above).  The explicit graph required in Theorem~\ref{t:polytimelist} is obtained from the disperser constructed by Ta-Shma, Umans, and Zuckerman~\cite{ats-uma-zuc:j:expanders}.

The lower bound in Theorem~\ref{th:quadraticlistt} is established via the equivalence between list approximability and graphs with on-line matching mentioned above. It is next observed that bipartite graphs capable of satisfying even  off-line matching need to have a certain expansion property and this imposes a lower bound on the left degree, which, as we have seen, corresponds to the list size.

The exponential lower bound in Theorem~\ref{th:computeSetsShortest} is shown using a connection with a type of Kolmogorov complexity that is less known, \emph{total conditional Kolmogorov complexity}, and then building strings with large such complexity using a game-theoretic approach. (The game-based technique in recursion theory was introduced by Lachlan~\cite{Lachlan} and 
further developed by A.Muchnik and others~\cite{KolmogorovGames,VerSurvey,ShenCie12}.)
\smallskip

\textbf{Paper organization.}
 The connection between list-approximability and graphs with on-line matching  is studied 
 in Section~\ref{matching}. \ver{We state there Theorems~\ref{th1a}, \ref{th2} and~\ref{th1b}
   about graphs with on-line matching and show that they imply
   Theorems~\ref{t:complist},
   \ref{t:polytimelist} and~\ref{th:quadraticlistt}, respectively.
   
The upper bounds, i.e.,
Theorem~\ref{th1a} (hence Theorem~\ref{t:complist}),
Theorem~\ref{th2} (hence Theorem~\ref{t:polytimelist})
and  Theorem~\ref{th:computeSetsShortest} (2) are proved in Section~\ref{s:ub}.
 The lower bounds, i.e., Theorem~\ref{th1b} (hence Theorem~\ref{th:quadraticlistt})}
 and Theorem~\ref{th:computeSetsShortest} (1), are proved in Section~\ref{s:lb}. 
In Section~\ref{s:other}, we observe that our technique  can  be used to improve Muchnik's 
Theorem~\cite{muc:j:condcomp} (see also the works of Musatov, Romashchenko and Shen~\cite{mus-rom-she:j:muchnik,mus:c:spacekolm,mus:c:spaceboundedextractor}), and a  result concerning  
distinguishing complexity of Buhrman, Fortnow, and Laplante~\cite{bfl:j:boundedkolmogorov}.

\section{List approximators for short programs and on-line matching}\label{matching}

We show that the problem of constructing approximators for short programs is equivalent to
constructing families of bipartite graphs which permit on-line matching, a notion 
 introduced in  a somewhat  different form in the paper of Musatov et al.~\cite{mus-rom-she:j:muchnik}.

Let a bipartite graph $G = (L,R,E\subseteq L\times R)$ be given,
where the set $L$ of left nodes and the set $R$ of right nodes consist of binary strings.
Assume that we receive ``requests for matching'' in the graph, each request having the form  (a binary string $x\in L$, a natural number $k$).
Such a request is satisfied  if we can assign to the left node $x$ a right neighbor  of length no larger than  $k$ plus a small overhead. 
For any  $x \in L$  it is possible to have several requests $(x, k_1), (x, k_2), \ldots $ and, thus, a left node $x$ may receive as matches  several right nodes. On the other hand, a  right node cannot match more than one left node.    For every $k$, there are at most $2^k$ requests of the form $(x,k)$ for various $x \in L$.  
Assignments cannot be revoked.  We will sometimes call right nodes \emph{hash-values}.

\begin{definition}\label{overhead}
Let $c(n)$ be a function of $n$
with natural values. A bipartite graph $G = (L, R, E\subseteq L\times R)$, whose left and right nodes are binary strings   
\emph{has matching with overhead $c(n)$} if
the following holds. For every set $S \subseteq L \times   \mathbb{N} $ of pairs  $(x, k)$
having at most $2^k$ pairs with the second component $k$ for all $k$, one can choose
for every pair $(x,k)$ in $S$ a neighbor $p(x,k)$ of $x$ so that   
$|p(x,k)|\le k+c(|x|)$ and $p(x_1,k_1)\ne p(x_2,k_2)$ whenever $x_1\ne x_2$. 
It is allowed that $p(x,k_1) = p(x,k_2)$ for some $x$.
If this is done, we say that $p(x,k)$ \emph{matches} $x$. 

A bipartite graph   
has \emph{on-line matching with overhead $c(n)$} 
if this can be done in the on-line fashion: requests for matching
$(x,k)$ appear one by one and we have to find $p(x,k)$ before the next 
request appears. All the made assignments cannot be changed.

A bipartite graph   
has \emph{computable} (respectively, \emph{polynomial time computable}) 
on-line matching with overhead $c(n)$ if it has online matching and the
matching strategy is computable (respectively, if the match for a left node $x$ can 
be found in time polynomial in $|x|$). 
We assume that the graph is available to the matching algorithm.
\end{definition} 
\begin{definition}
A bipartite graph is \emph{ computable (polynomial-time computable)} if given a left node~$x$ we can compute (respectively, compute in  time polynomial in $|x|$)
the list of all its neighbors. A polynomial-time computable graph is also said to be \emph{explicit}.
\end{definition}

The next two theorems show that for any function $c$, the existence of a
computable (polynomial-time computable)  list approximator $f$ for
$c(|x|)$-short programs is equivalent the existence of a computable
(polynomial-time computable) graph $G$ that has on-line matching with overhead
$c(|x|) + \bigO{1}$ and $|f(x)|$ is equal to the degree of $x$ in $G$; 
up to an assumption on the computability of the matching strategy.
\begin{theorem}[Graph $G$ with on-line matching $\Rightarrow$ list-approximator $f$]\label{th3}
Assume there is a  computable  
graph with $L=\zo^*$ where each left node $x$ has degree $D(x)$ and 
which has on-line matching with overhead $c(n)$. 
Assume further that the matching strategy is computable. 
Then for every standard machine $U$ there exists a computable function $f$ that for any $x$ produces a
list with  $D(x)$ many elements containing a program $p$ for $x$ of length $|p| = C_U(x) + c(|x|)+\bigO{1}$.
If the graph is polynomial time computable then the function $f$ is polynomial time computable, too.
\end{theorem}

\begin{proof}
Recall that for polynomial time results, we assume that a machine is standard through a
function that is computable in polynomial time.

Run the optimal machine  $U(q)$ in parallel for all strings $q$.
Once $U(q)$ halts with the result $x$  
we pass the request $(x,|q|)$ to the matching
algorithm in the graph and find a hash value $p$ of length at most $|q|+c(|x|)$ for $x$. 

By construction, every string $x$ is matched to a string $p$ of length
at most $C_U(x)+c(|x|)$. Each right node is matched to at most one node in the graph. 
Hence there is a machine $V$ such that $V(p)=x$ whenever $p$ is matched to $x$. 
Thus for every string  $x$ there is a neighbor $p$ of  $x$ with
$|p|\le C_U(x)+c(|x|)$ and $V(p)=x$. As $U$ is a standard machine, there is a (polynomial time) computable
function $t$ with $U(t(p))=V(p)$ and $|t(p)|\le|p|+\bigO{1}$. 
Let $f(x)$ be the list consisting of  $t(p)$ for all  the neighbors $p$ of $x$ in the graph. By  construction
$|f(x)|=D(x)$ and we are done.  
\end{proof}


\begin{theorem}[List-approximator $f$ $\Rightarrow$  graph $G$ with on-line matching]
\label{th3a}
Assume that $c(n)$ is (polynomial-time) computable function and 
there are an optimal machine $U$ 
and a (polynomial-time) computable function $f$ that for any $x$ produces a
finite list containing a program $p$ for $x$ of length $|p| \le C_U(x) + c(|x|)$.
Consider the bipartite graph $G$ with $L=\zo^*$
where the neighbors of node $x$ are all strings from $f(x)$.
Then $G$  has on-line  matching with overhead $c(|x|)+\bigO{1}$.
\end{theorem}
\begin{proof} 
For each $n$ let $G_n$ be the subgraph of $G$ with $L=\zo^{\le n}$ and $R$ the set of neighbors of $L$.
Without loss of generality, we assume that all strings in $f(x)$ have length at most $|x|+\bigO{1}$
and hence the graph $G_n$ is finite.  
We claim that $G_n$ has on-line matching with overhead $c(|x|)+\bigO{1}$ for all $n$,
(where the $\bigO{1}$ constant does not depend on $n$). 

We first show that this implies the theorem.
Suppose that $M_1, M_2, \dots$ are on-line matching strategies for graphs $G_1, G_2, \dots$
It suffices to convert them to strategies $M'_1, M'_2, \dots$ for $G_1, G_2, \dots$
such that for all $i$, $j > i$ strategy $M'_j$ is an extension of $M'_i$, 
i.e. on a series of requests only containing nodes from $G_i$, 
strategy $M'_j$ behaves exactly as $M'_i$. 
Because each $G_n$ is finite, there are only finitely many different matching strategies for $G_n$.
Hence, there is a strategy $M'_1$ that equals the restriction of $M_n$ to $G_1$ for infinitely many $n$.
Therefore there is also a strategy $M'_2$ that is an extension of $M'_1$ and 
equals the restriction of $M_n$ to $G_2$ infinitely often, and so on. 

It remains to show the claim.
For the sake of contradiction assume that for every constant $i$ there is $n$
such that $G_n$ does not have  on-line matching with overhead  $c(|x|)+i$, 
and $f$ is a list-approximator for $c(|x|)$-short programs on $U$. 
Because $G_n$ is finite, for all $n$ and $c$ one can find
algorithmically (using an exhaustive search) whether $G_n$ has on-line matching  
with overhead  $c(|x|)+i$ or not. One can also find a winning strategy for that player
who wins (``Matcher'' or ``Requester'').
Therefore for every $i$ we can algorithmically find the first $n$
such that the graph $G_n$ does  not have on-line matching with overhead  $c(|x|)+i$ 
and the corresponding winning strategy for Requester for $G_n$.

Let that strategy play against the following ``blind'' strategy of Matcher.
Receiving a request $(x,k)$,  the Matcher runs $U(p)$ for all $p\in f(x)$,  $|p|\le k+c(|x|)+i$,  in parallel. 
If for some $p$, $U(p)$ halts with the result $x$, he matches the first such $p$ to $x$ and proceeds to the next
request. Otherwise the request remains not fulfilled. 

Consider the following machine $V$. On input $(q,i)$, with $q$ a $k$-bit long string and $i$ a natural number,  it finds 
the first $G_n$ such that the graph $G_n$ does  not have  on-line matching with overhead  $c(|x|)+i$ and a winning strategy for Requester,
and runs it against the blind strategy of the  Matcher. Then it returns $x$, where $(x,k)$ is the $q$th 
request with the second component $k$  (we interpret the string  $q$ as the ordinal of the request in some  standard manner).
Since the Requester wins, there is a request $(x,k)$ that 
was not fulfilled. 
We have 
\begin{equation}\label{eq34}
C_U(x)\le C_V(x)+\bigO{1}\le k+ 2\log i+\bigO{1}\le k+ i;
\end{equation}
the last inequality holds for all large enough $i$.
As the request $(x,k)$  
was not fulfilled, there is no $p$ in $f(x)$ with $|p|\le k+i+c(|x|)$. Due to~\eqref{eq34},
$f(x)$ has no $c(|x|)$-short program for $x$, a contradiction. 
\end{proof}

We have reduced the problem of  list approximation for $c$-short programs  to the construction of bipartite graphs with on-line matching that have polynomial left degree.
Our main technical contributions are the following theorems.

\begin{theorem}[Combinatorial version of Theorem~\ref{t:complist}]\label{th1a}
There is a computable graph with $L=\zo^*$ and left degree $D(x) = |x|^2$ 
which has polynomial time on-line matching with overhead $\bigO{1}$. 
\end{theorem}

\begin{theorem}[Combinatorial version of Theorem~\ref{t:polytimelist}]\label{th2}
There is a polynomial time computable  
graph with $L=\zo^*$ with left degree $D(x) = \poly(|x|)$ 
which has polynomial time on-line matching with overhead $\bigO{\log |x|}$. 
\end{theorem}

\begin{theorem}[Combinatorial version of Theorem~\ref{th:quadraticlistt}]\label{th1b}
In every graph $G$ with $L=\zo^n$ 
that has off-line matching with overhead $c$, 
the maximal degree of left nodes is $\Omega(n^2/c^2)$. 
\end{theorem}

Theorems~\ref{th1a}, \ref{th2}  and~\ref{th1b} imply Theorems~\ref{t:complist},
\ref{t:polytimelist} and~\ref{th:quadraticlistt}, respectively. 

Our on-line matching strategy used in the proofs of Theorems~\ref{th1a} and~\ref{th2} are very simple:
receiving a new request of the form $(x,k)$ we just find the maximal $i\le k+c$ such that 
there is a free (\ie, not yet used in a matching)  right neighbor of $x$ of length $i$ and match  $x$ with the first such neighbor of $x$.

\section{The upper bounds}
\label{s:ub}
In this section we prove Theorem~\ref{th1a}, Theorem~\ref{th:computeSetsShortest}(2) and Theorem~\ref{th2}. Essentially, the proofs consist in  the construction of the corresponding graphs that have on-line matching with overhead $c(n)$. 
These are infinite graphs which are obtained as the union of finite bipartite graphs with $L = \zo^n$ and $R=\zo^{k+c(n)}$. For such graphs in which the length of the right nodes is fixed we do not need to refer to the
length constraint in a matching request as in Definition~\ref{overhead}, and we can work with the following simpler definition.

\begin{definition}\label{def2}
A  bipartite graph has \emph{matching up to $K$ with at most $M$ rejections},
if for any set of left nodes of size at most $K$ we can drop at most $M$ of its elements 
so that there is a matching in the graph for the set of remaining nodes.
A graph has  an \emph{on-line matching up to $K$ with at most $M$ rejections} 
if we can do this in on-line fashion.
For $M=0$ we say that the graph has \emph{matching up to $K$}.
\end{definition}
The connection between graphs with (on-line) matching up to $K$ (Definition~\ref{def2})  and graphs
with (on-line)  matching with overhead $c$ (Definition~\ref{overhead}) is the following. 
If a graph $G$ has  (on-line) matching with overhead $c$ then removing from $G$ all left nodes of length different from $n$ and 
all right nodes of length more than $k+c(n)$ we obtain a graph with (on-line) matching up to $2^k$. 
On the other hand, assume that, for some  $n$, for all $k<n$,  we have a graph $G_{n,k}$ with $L=\zo^n$  and $R=\zo^{k+c(n)}$ which
has (on-line) matching up to $2^k$. Then the union $G_n$ of $G_{n,k}$ over all $k<n$ has (on-line)
matching with overhead $c$, provided  all requests $(x,k)$ satisfy $k<|x|$. At the expense of increasing
the degree and $c$ by 1, the graph $G_n$ can be easily modified to have (on-line)
matching with overhead $c$ unconditionally: append 0 to all right nodes of $G_n$ and for every 
$x\in\zo^n$ add a new right node $x1$ connected to $x$ only. 

Thus, we will need to construct finite graphs as in Definition~\ref{def2}. This will be done by constructing a certain type of expander graphs.
\begin{definition}\label{d:expander}
A bipartite graph is called a \emph{$(K,K')$-expander},
if every set of $K$ left nodes has at least $K'$ distinct neighbors. 
\end{definition}
Graphs that have off-line matching up to $K$ are closely related to $(K,K)$-expanders. Indeed,
any graph having  matching up to $K$ is obviously a  $(K',K')$-expander
for all $K'\le K$. Conversely, by Hall's theorem~\cite{hall},  any graph which is a  $(K',K')$-expander
for all $K'\le K$ has  matching up to $K$.

In Musatov et al.'s~\cite{mus-rom-she:j:muchnik} it was shown that a reduction from expanders to \emph{on-line}
matching is also possible. More specifically,
every family of $(2^i,2^i)$-expanders, one for each  $i<k$, sharing the same set $L$ of left 
nodes can be converted into a graph with the same set $L$ of left nodes that has on-line
matching up to $2^k$, at the expense of multiplying the degree by $k$ and increasing hash-values by 1.  
(We will present the construction in the proof of Theorem~\ref{th7}.)

In Musatov et al.'s paper~\cite{mus-rom-she:j:muchnik} it is observed that every $(K,K)$-expander
has on-line matching up to $2K$ with at most $K$ rejections. We need a slight generalization of  this fact.

\begin{lemma}
Every $(M,K-M)$-expander
has on-line matching up to $K$ with at most $M$ rejections.
\end{lemma}
\begin{proof}
Use the following greedy  strategy
for on-line matching: each time a left vertex is received, check if it has a neighbor 
that was not used yet. If yes, any such neighbor is selected as the match for that node.
Otherwise, the node is rejected.  

For the sake of contradiction, assume that the number of rejected nodes is more
than $M$. Choose from them exactly $M$ rejected nodes. 
By the expansion property, they have at least $K-M$ neighbors and 
all those neighbors are used by the greedy strategy (otherwise the node having a non-used neighbor
would not be rejected). Thus we have at least $K-M$ matched left nodes
and more than $M$ rejected nodes. Consequently, we have received more than $K$ requests, a contradiction.
\end{proof}

For Theorem~\ref{th1a} we will use non-explicit such expander graphs, for Theorem~\ref{th2} we will need
explicit such graphs, which we obtain from the disperser of Ta-Shma, Umans, and Vadhan~\cite{ats-uma-zuc:j:expanders}.
We say that a family of graphs $G_{n,k}$ is computable (respectively, computable in polynomial time) 
if given $n,k$ and a left node $x$ in $G_{n,k}$, we can compute (respectively, compute in polynomial time)
the list of all neighbors of $x$ in $G_{n,k}$.

\begin{theorem}\label{th7}
  Given a computable (respectively computable in polynomial time) family
  $G_{n,k}$ of $(2^k,2^k)$-expanders with $L=\zo^n$, $R=\zo^{k+c(n)}$ and the degree of all left nodes is at most $D(n)$, 
  we can construct a computable graph $G$
  with $L=\zo^*$ that has a \ver{computable (respectively computable in polynomial time)} on-line matching with overhead $c(n)+\bigO{\log n}$
  and the degree of each left node is $\bigO{D(n)n}$.  
\end{theorem}
\begin{proof}
The main tool is borrowed from the paper of Musatov et al.~\cite{mus-rom-she:j:muchnik}:
all the graphs $G_{n,k}$ share the same set of left nodes while their sets of right nodes
are disjoint. Let $H_{n,k}$ denote the union of $G_{n,i}$ over all $i<k$.
Then  $H_{n,k}$ has on-line matching up to $2^k$ (without
rejections).  Indeed, each input left node is first given to the matching
algorithm for $G_{n,k-1}$ (that has on-line matching up to $2^k$ with at most $2^{k-1}$ rejections)
and, if rejected is given to the matching algorithm for 
$G_{n,k-2}$ and so on. 

Using this construction we can prove the theorem with slightly worse parameters as claimed.
To this end identify right nodes of the graph $H_{n,k}$ with strings of length $k+c(n)$
(the number of right nodes of  $H_{n,k}$ does not exceed the sum of geometrical series $2^{k+c(n)-1}+2^{k+c(n)-2}+\dots<2^{k+c(n)}$).
The left degree of $H_{n,k}$ is $D(n)k$. 

Recall the connection between matching up to $2^k$ and matching with overhead (the  paragraph
after Definition~\ref{def2}).
We see that the family $H_{n,k}$ can be converted into 
a graph $H_n$ with $L=\zo^n$ and degree $\ver{D'(n)=}D(n)n(n-1)/2+1$ having on-line matching  with overhead $c(n)+1$.
Finally, we prefix the  right nodes of $H_n$  with the  $\bigO{\log n}$-bit prefix-free code of the number $n$
and consider the union of all $H_n$. The resulting graph has on-line matching 
with overhead $c(n)+\bigO{\log n}$, its set of left nodes is $\zo^*$ and the degree
of every left node of length $n$ is $\bigO{D(n}n^2)$.

Now we will explain how to reduce the degree to $\bigO{D(n)n}$.
Consider four copies of $G_{n,k-1}$  with the same set $L$ of left nodes 
and disjoint sets of right nodes (say append 00 to every right node to get the first 
copy, 01 to get the second copy and so on). Their union is a 
$(2^{k-1},2^{k+1})$-expander, and hence has matching up to $2^{k+1}$ with at most $2^{k-1}$ rejections.\footnote{One can also
consider the union of $G_{n,k-1}$ and $G_{n,k}$, which also   has matching up to $2^{k+1}$ with at most $2^{k-1}$ rejections.
}
Its left degree is $4D(n)$ and the length of right nodes is $k+c(n)+2$.
Replace in the above construction of $H_n$ the graph $H_{n,k}$ by this graph. 
Thus the left degree of $H_n$ becomes $\bigO{D(n)n}$ in place of $\bigO{D(n)n^2}$. 
It remains to show that (the union of all graphs) $H_n$ has still on-line matching with overhead  $c(n)+\bigO{\log n}$ 

Again the matching strategy is greedy. Once we receive a request $(x,k)$ with $|x|=n$, we match 
$x$ to $1x$ if $k\ge n$. Otherwise,  we pass $x$ to the matching algorithm
in  $H_{n,k}$. 
If the algorithm rejects $x$, we pass $x$  to the matching
algorithm in $H_{n,k-1}$ and so on. 
We claim that we eventually find a match in one of 
the graphs $H_{n,i}$ for $i\le k$.
To prove the claim  it suffices to show that the 
matching algorithm for $H_{n,k}$ receives 
at most $2^{k+1}$ input strings. This is proved by a downward induction on $k$ (for any fixed $n$).
For the base case,  $k=n-1$,  this is obvious: we try to match 
in $H_{n,n-1}$ up to $2^{n-1}$ strings.
The induction step: by induction hypothesis the matching algorithm for 
$H_{n,k+1}$ receives at most $2^{k+2}$ input strings an thus rejects at most $2^{k}$ of them.
The matching algorithm  for $H_{n,k}$ thus receives at most $2^{k}$ rejected strings and at most $2^{k}$ new ones,
coming from requests of the form $(x,k)$. 
\end{proof}

\subsection{Proof of Theorem~\ref{th1a}}
\label{s:complist}

A weaker form of Theorem~\ref{th1a} can be derived from Theorem~\ref{th7}
and the following lemma of Muchnik~\cite{muc:j:condcomp}.

\begin{lemma} \label{l:onlineMatching}
For all $n$ and $k<n$, there exists a  $(2^{k},2^{k})$-expander with $L=\zo^n$,  $R=\zo^{k+2}$  and  
all left nodes have degree at most $n+1$.
\end{lemma}

\begin{proof}
We use the probabilistic method, and for each left node we choose its $n+1$ neighbors at random:
all neighbors of each node are selected independently 
among all $2^{k+2}$ right nodes with uniform distribution,
and the choices for different left nodes
are independent too.   
We show that  the expansion property is satisfied with positive probability. Hence there exists at least one
such graph. To estimate the probability that the property is not satisfied,
consider a pair of sets  
$L'$ and $R'$ of left and right nodes, respectively, of sizes $2^{k},2^{k}-1$.
The probability that the neighbors of all nodes in $L'$ belong to $R'$ is upper-bounded by $(1/4)^{(n+1)2^k}$.
The total probability that the expansion condition is not satisfied, is obtained by summing over all such $L',R'$, i.e.
\[
\begin{array}{ll}
\left(\frac{1}{4}\right)^{(n+1)2^{k}} (2^n)^{2^{k}} (2^{k+2})^{2^{k}-1} &\le
\left(\frac{2^n2^{k+2}}{4^{n+1}}\right)^{2^{k}} \\
& \le  \left(\frac{2^n2^{n+1}}{4^{n+1}}\right)^{2^{k}} \\
& =
\left(\frac{1}{2}\right)^{2^{k}} \\
& <1.
\end{array}
\]
\renewcommand{\qed}{}
\end{proof}

\begin{remark}\label{remark:xxx}
By the very same construction we can obtain 
a graph with $L=\zo^n$,  $R = \zo^{k+2}$,  $D=n+1$
that is a $(t,t)$-expander \emph{for all} $t\le 2^{k}$.
Indeed, the probability that a random graph is not a $(t,t)$-expander 
is at most $\left(\frac{1}{2}\right)^{t}$ (we may replace $2^k$ by $t$ in the above 
formulas). By the union bound, the probability that this happens for some $t\le 2^k$ is at most
the sum of the  geometric series   $\sum_{t=1}^{2^k}\left(\frac{1}{2}\right)^{t}<1$.
By Hall's theorem, this graph has off-line matching up to $2^k$.
An interesting open question is whether there is a graph with the same parameters, \ie,  $L=\zo^n$,  $R = \zo^{k+\bigO{1}}$,
$D=\bigO{n}$, that has  \emph{on-line} matching up to $2^k$.
\end{remark}

From lemma~\ref{l:onlineMatching} and Theorem~\ref{th7} we obtain a computable 
graph with on-line matching with overhead  $\bigO{\log |x|}$, degree $\bigO{|x|^2}$ and $L=\zo^*$. 
We now need to replace the $\bigO{\log |x|}$ overhead by $\bigO{1}$. 
Recall that the $\bigO{\log |x|}$ appeared from the prefix code of
$|x|$ added to the hash values. To get rid of it we need a computable graph
$F_k$ in place of the previously used $G_{n,k}$ with the same parameters but with $L=\zo^{\geq k}$,
and not $L=\zo^n$.  Such a graph is constructed in the following lemma.

\begin{lemma}
\label{l:compositeMatching}
For every $k$ there is a computable bipartite graph $F_{k}$ with $L=\zo^{\ge k+3}$, $R=\zo^{k+12}$
that is a $(2^k,2^k)$-expander and  the degree of every left node $x$ is $|x|/4$.
\end{lemma}

\begin{proof} 
  We first build such a graph with left nodes being all strings of length between $k+3$ and $K=2^{k+2}$. 
This is again done by the probabilistic method: we choose $|x|/4$ neighbors of every node $x$ independently. 
Let  $L_i$ stand for all left nodes of
length $i$. For any $i \in k+3, \dots, K$, the probability that all elements  of a fixed $L' \subset L_i$ 
are mapped to a fixed set of size at most $2^k-1$ 
at the right is at most $\left(\frac{1}{2^{12}}\right)^{i|L'|/4}$.
The probability that some $t_i$ elements 
in $L_i$ are mapped into a fixed set of $2^k-1$ elements at the right is at most
\[
\begin{array}{ll}
2^{it_i} \left( \frac{1}{2^{12}} \right)^{it_i/4} & = \left( \frac{1}{2} \right)^{2it_i} \\
& \le \left( \frac{1}{2} \right)^{( k+3)2t_i}  \\
& = \left( \frac{1}{2K} \right)^{2t_i}.
\end{array}
\]
If $\sum_{i=k}^K t_i = t$ with $t = 2^k$, 
the probability that the union of neighbors of 
$t_k$ elements in $L_k$, $t_{k+1}$ elements in $L_{k+1}$, \dots , and $t_K$ elements in $L_K$ 
are mapped to a fixed set of size at most $2^k-1$ is bounded by
$\prod_i\left( \frac{1}{2K} \right)^{2t_i}=\left( \frac{1}{2K} \right)^{2t}$.
Multiplying by the number $K^{2^k-1} \le K^{t}$ of different right sets of size $K-1$, 
and multiplying by the upper bound  $K^t$ for the number of  solutions to the equation $\sum_{i=k}^K t_i = K$,
we find 
\[
\left( \frac{1}{2K} \right)^{2t} K^{t} \; K^{t} \le \left( \frac{1}{4} \right)^t < 1.
\]
Hence, the total probability to randomly generate a graph that is not an expander is strictly less than $1$. 
Therefore, a graph satisfying the conditions must exist, and can be found by exhaustive search. 

On the left side, we now need to add the strings of length larger than $K = 2^{k+2}$. 
These nodes are connected to all
the  nodes on the right side. Thus the degree of every such node $x$  is $2^{k} \le |x|/2^{2}$  
and we are done.
\end{proof}

\begin{remark}
By the very same construction we can obtain 
a graph with $L=\zo^{\ge k}$, $R=\zo^{k+12}$,  $D=\bigO{n}$
that is a $(t,t)$-expander \emph{for all} $t\le 2^{k}$
and thus has off-line matching up to $2^k$ 
(use the union bound over all $t\in\{k,\dots,K\}$).
An interesting open question is whether there is a graph with the same parameters that has  \emph{on-line} matching up to $2^k$.
\end{remark}

\begin{proof}[Proof of Theorem~\ref{th1a}.]
Appending all 2-bit strings to all the right nodes of the graph $F_{k-1}$ (and thus increasing the degree 4 times)
we obtain a $(2^{k-1},2^{k+1})$-expander $H_{k}$. The union of $H_{k}$ over all $k$ is a computable graph, whose 
left degree is $|x|(|x|-3) \le |x|^2-1$,
and the set of left nodes is $\zo^*$. 
For each left node $x$, we \ver{add} an additional node to handle requests of the form $(x,k)$ with $k > |x|-3$.
This graph has on-line matching with constant overhead.
This is proved by  downward induction, as in Theorem~\ref{th7}. Indeed, consider
the step when the $s$th request for matching arrives. 
By a downward induction on $k$ we can again prove that the number of 
matching requests in $H_k$ is at most $2^{k+1}$. Now the base of induction  
is the maximal $k$ for which
there has been at least one request for matching in $H_k$  among the $s$ requests so far.
We conclude that the $s$th request  is satisfied and, since this holds for every~$s$, we are done. 
\end{proof}

\ver{As we have seen, 
Theorem~\ref{th1a} implies Theorem~\ref{t:complist}.}

\ver{\subsection{Proof of Theorem~\ref{th:computeSetsShortest}(2)}
By Theorem~\ref{th1a} there is a computable graph $G$ having a computable on-line matching strategy
with constant overhead $c$ and degree $|x|^2$.

Fix any standard machine $U$.
In the proof of Theorem~\ref{th3} we have constructed a computable function 
$V$ such that $C_V(x)\le C_U(x)+c$ and  such that \emph{every} $V$-program of
every string $x$ is a neighbor of $x$ in the graph $G$.
If it happens that $V$ is a standard machine, then we are done: 
consider the list $\{p\mid p$ is a neighbor
of $x\}$.


Otherwise, define the machine $U_1$
by letting $U_1(0p)=V(p)$ and $U_1(1^{c+2}p)=U(p)$.
The second equation guarantees that $U_1$ is
a standard machine. Both equations
imply that for every $x$ the
0-shortest $U_1$-program for $x$ has the form
$0p$ (recall that $C_V(x)\le C_U(x)+c$). Hence for all $x$
the 0-shortest
$U_1$-program $q$ for $x$ in the list $\{0p\mid p$ is a neighbor
of $x\}$.
}

\subsection{ Proof of  Theorem~\ref{th2}}

By Theorem~\ref{th7},   we have to construct for every $k\le n$
an explicit $(2^k,2^k)$-expander of left degree $\poly(n)$, with $2^n$ left nodes and $\poly(n)2^k$ right nodes.
(Recall that a graph is \emph{explicit} if there is an algorithm that on input $x \in \zo^n = L$ lists in $\poly(n)$ time  all the neighbors of $x$.)

The proof relies on the explicit disperser graphs of  Ta-Shma, Umans, and Zuckerman from Theorem~\ref{t:tuz} below.
\begin{definition}
A bipartite graph $G = (L,R,E)$ is a \emph{$(K, \delta)$-disperser}, if every subset  
$B \subseteq L$ with $|B| \geq K$  has at least $(1-\delta)|R|$ distinct neighbors.
\end{definition}
\begin{theorem}
\label{t:tuz} [Ta-Shma, Umans, Zuckerman~\cite{ats-uma-zuc:j:expanders}] 
For every $K,n$ and constant $\delta$,  
there exists an  explicit $(K, \delta)$-disperser
$G = (L=\zo^n, R= \zo^m, E \subseteq L \times R)$ in which every node in $L$ has degree 
$D=\poly(n)$ and \ver{$|R|=\frac{ \alpha K D}{n^3}$}, for some constant $\alpha$.
\end{theorem}

Given $n$ and $k$,  we  apply this theorem to $K=2^k$ and $\delta=1/2$.
We obtain a $(2^k, \frac{\alpha 2^k D}{2n^3})$-expander with degree $D=\poly(n)$, $L=\zo^n$  and
\ver{$|R|=\frac{\alpha K D}{n^3}$}.
Consider $t=\max\{1,\lceil \frac{2n^3}{\alpha D}\rceil\}$ disjoint copies of this graph and identify left
nodes of the resulting graphs (keeping their sets of right nodes disjoint). 
We get an explicit  $(2^k, 2^k)$-expander with $2^n$ left and $2^{k}\poly(n)$ right nodes
and degree $\poly(n)t=\poly(n)$. 
\smallskip

\ver{As we have seen, Theorem~\ref{th2} implies Theorem~\ref{t:polytimelist}.}

\section{The lower bounds}\label{s:lb}

\subsection{Proof of Theorem~\ref{th1b}}
Assume that $G$ has off-line matching with overhead  $c$.
Let $G[\ell,k]$ denote the induced graph that is obtained from $G$ 
by removing all right nodes of length more than $k$ and less than $\ell$.
The graph $G[0,k+c]$ is obviously a $(2^k,2^k)$-expander for every $k$.
As there are less than $2^{k-1}$ strings of length less than $k-1$,
it follows that the graph   $G[k-1,k+c]$ is a $(2^k,2^{k-1}+1)$-expander.

The next lemma, inspired by Kov\H{a}ri, S\`{o}s and Tur\`{a}n\cite{kovari1954} (see Radhakrishnan and Ta-Shma~\cite[Theorem 1.5]{rad-tas:j:extractors}),
shows that any such expander must have large degree.

\begin{lemma}\label{l1}
Assume that a bipartite graph with $2^\ell$ left nodes and $2^{k+c}$ right nodes 
is a $(2^k,2^{k-1}+1)$-expander. Then there is a left 
node in the graph with degree more than $D=\min\{2^{k-2}, (\ell-k)/(c+2)\}$.
\end{lemma}
\begin{proof}
For the sake of contradiction assume that all left nodes have degree
at most $D$ (and without loss of generality we may assume that all degrees are exactly $D$). 
We need to find a set of right nodes $B$ of size $2^{k-1}$ 
and $2^k$ left nodes all of whose neighbors lie in $B$.
The set $B$ is constructed via a probabilistic construction. Namely, choose $B$ at random (all $\binom{2^{k+c}}{2^{k-1}}$
sets have equal probabilities). The probability that a fixed neighbor of a fixed left node is in $B$
is equal to 
$$
\frac{\binom{2^{k+c}-D}{2^{k-1}-D}}{\binom{2^{k+c}}{2^{k-1}}}=
\frac{2^{k-1}(2^{k-1}-1)\cdots(2^{k-1}-D+1)}{2^{k+c}(2^{k+c}-1)\cdots(2^{k+c}-D+1)}.
$$
Both products in the numerator and denominator have $D$ factors and the ratio of corresponding factors is
at least 
$$
\frac{2^{k-1}-D+1}{2^{k+c}-D+1}\ge 2^{-c-2}
$$ 
(the last inequality is due to the assumption $D\le 2^{k-2}$).
Thus the probability that all neighbors of a fixed left node are in $B$ is at least 
$2^{-D(c+2)}$. Hence the average number of left nodes having this property is at least 
$2^{\ell-D(c+2)}$, which is greater than or equal to $2^k$ by the choice of $D$. 
Hence there is $B$ that includes neighborhoods of at least $2^k$ left nodes,
a contradiction. 
\end{proof}

This lemma states that at least one left node has large degree. However, it implies more: 
if the number of left nodes is much larger than $2^\ell$, then  
almost all left nodes must have large degree. Indeed, assume that a bipartite graph with 
$2^{k+c}$ right nodes is a $(2^k,2^{k-1}+1)$-expander.
Choose $2^\ell$ left nodes with smallest degree and apply the lemma
to the resulting induced graph (which is also  a $(2^k,2^{k-1}+1)$-expander).
By the lemma, in the original graph all except for less than $2^\ell$ nodes have degree 
more than $D=\min\{2^{k-2}, (\ell-k)/(c+2)\}$. 

\begin{proof}[Proof of \ver{Theorem}~\ref{th1b}.]
Choose $n/4<k\le n/2$. As noticed,  
the graph $G[k-1,k+c]$ is a $(2^k,2^{k-1}+1)$-expander and has less than $2^{k+c+1}$ right nodes.
By Lemma~\ref{l1} and the above observation (applied to $\ell=3n/4$), all except for at most $2^{3n/4}$
left nodes of $G[k-1,k+c]$  have degree at least $n/(4(c+3))$. 

Pick now $\ell$ different integers $k_i$ with $n/4 < k_1 < k_2 < \ldots  < k_\ell < n/2$ that are $c+2$ apart of each other, where $\ell$ is about  $n/(4(c+2))$. 
In each graph $G[k_i - 1, k_i + c]$, all  the left nodes, except at most $2^{3n/4}$, have degree $\geq  n/(4(c+3))$.  Since the $k_i$ are $c+2$ apart, and thus the graphs 
$G[k_i - 1, k_i + c]$ have
pairwise disjoint right sets, it follows that $G$ has left nodes with degree $\Omega(n^2/(c+3)^2)$.
\end{proof}

\ver{As we have seen, Theorem~\ref{th1b} implies Theorem~\ref{th:quadraticlistt}.}

\subsection{Proof of Theorem~\ref{th:computeSetsShortest} (1).}

The size of list-approximators is closely related to total conditional Kolmogorov complexity, which was first introduced by A.~Muchnik
and used by Bauwens and Vereshchagin~\cite{BauwensPhd,Ver09}.  {\em Total conditional Kolmogorov complexity} with respect to $U$ is defined as:
\[
\CT_U(u|v) = \min \left\{ |q| : U(q,v) = u \wedge \forall z\left[ U(q,z) \text{ halts} \right] \right\}. 
\]
If $U$ is a standard machine then $\CT_U(u|v)\le\CT_V(u|v)+c_V$ for every machine $V$. 
The connection to list-approximators is the following:
\begin{lemma}
If $f$ is a computable function that maps every string to a finite list of strings
then $\CT_U(p|x)\le \log|f(x)|+\bigO{1}$ for any standard machine $U$ and every $p$ in $f(x)$. The constant in $O$-notation
depends on $f$ and $U$.
\end{lemma} 
\begin{proof}
Let $V(j,x)$ stand for the $j$th entry of the list $f(x)$, if $j\le  |f(x)|$, and 
for the  (say) empty string otherwise. Obviously $\CT_V(p|x)\le \log|f(x)|$ for all $p$ in $f(x)$. Hence 
$\CT_U(p|x)\le \log|f(x)|+\bigO{1}$.  
\end{proof}

Thus to prove Theorem~\ref{th:computeSetsShortest}(1) it suffices to construct for every $c$  
a standard machine $U_c$ such that,  
for infinitely many $x$,  every $c$-short $p$ for $x$ with respect to $U_c$ satisfies $\CT_U(p|x)\ge |x|-\bigO{1}$.

We first consider the case $c=0$. Fix a standard machine $U$. We construct another machine $V$, a constant $d$ 
and a sequence of pairs of strings $(x_1, p_1), (x_2, p_2), \dots$ such that \\
(a) $p_k$ is the unique 0-short program for $x_k$ with respect to $V$,\\
(b) $C_U(x_k)\ge k$,\\
(c) $|x_k|=|p_k|=k+d$,\\
(d) $\CT_U(0p_k|x_k)\ge k$.

Once such $V$ has been constructed, we let 
$U_0(0q)=V(q)$ and $U_0(1^{d+2}q)=U(q)$. 
The latter equality guarantees that $U_0$ is a standard machine. And both equalities 
together with items (a), (b) and (c) imply that $0p_k$ is the unique 0-short program for $x_k$
with respect to $U_0$. Finally, item (d) guarantees that its total complexity conditional to $x_k$ 
is at least $|x_k|-d$.
The construction of $V, d$ and $(x_k,p_k)$ can be described in game terms. 
\smallskip

 \textbf{Description of the game.}
The game has integer  parameters $k,d$ and is played on a rectangular grid with $2^{k+d}$ rows and $2^{k+d}$ columns.
The rows and columns are identified with strings of length $k+d$. 
Two players, Bob and Alice, play in turn. In her turn Alice 
can either pass or put a token on the board. Alice can place at most one token in each row
and at most one token in each column. Once a token is placed, it can not be moved nor removed. 
In his turn Bob can either pass, or choose a column and \emph{disable} all its cells, or 
choose at most one cell in every column and \emph{disable} all of them. 
If a player does not pass, we say that she/he \emph{makes a move}.
Bob should make less than $2^{k+1}$ moves. 
The game is played for an infinite time and 
Alice loses if at some point after her turn, all her tokens are in disabled cells. 

We will show that, for $d=3$, for every $k$,  Alice wins this game. More specifically, there is a winning strategy 
for Alice that is uniformly computable given $k$. Assume that this is done. Then consider the following ``blind''
strategy for Bob: start an enumeration of all strings $x$ with $C_U(x)<k$ and all strings $q$ of length less than 
$k$ such that $U(q,x)$ halts for all $x$ of length $k+d$.
That enumeration can be done uniformly in $k$. In his $t$th turn Bob:
disables all cells in the $x$th column, if 
on step $t$ in this enumeration a new $x$ of length $k+d$  with $C_U(x)<k$ appears;
disables all cells $(p,x)$ with $|x|=|p|=k+d$, $U(q,x)=0p$, if
on step $t$  a new string $q$ of length less than 
$k$ appears such that $U(q,x)$ halts for all $x$ of length $k+d$; 
and passes if none of these events occurs.
Note that the total number of Bob's moves is less than $2^k+2^k=2^{k+1}$, as required.

Now consider the following machine $V(p)$: let $k=|p|-d$ and let  
the Alice's computable winning strategy play against Bob's blind strategy.
Watch the play waiting until Alice places a token on a cell $(p,x)$ in the $p$th row. Then output 
$x$ and halt. Note that such $x$ is unique (if exists), as Alice places at most one token in each row.
Because Alice's strategy is winning, at least one of the finitely many tokens in a cell $(p,x)$ will never be disabled. This pair $(p,x)$ satisfies all the requirements (a)--(d).
Thus it suffices to design a computable winning strategy for Alice.

 \textbf{A winning Alice's strategy.}
The strategy is a greedy one. In the first round Alice places a token in any cell. Then she waits until that cell becomes 
disabled. Then she places the second token in any enabled cell that lies in another row and another column and
again waits until that cell becomes disabled. At any time she chooses any enabled cell 
that lies in a row and a column that both are free of tokens. In order to show that Alice wins, we just
need to prove that there is such a cell. Indeed, Bob makes less than 
$2^{k+1}$ moves, thus Alice makes at most $2^{k+1}$ moves. On each of Bob's moves at most 
$2^{k+d}$ cells become disabled.  On each of Alice's moves at most $2^{k+d+1}$ cells becomes non-free because 
either their column or row already has a token.
Thus if the total number of cells is more than 
$$
2^{k+1}2^{k+d}+2^{k+1}2^{k+d+1}=6\cdot2^{2k}2^d, 
$$
we are done. The total number of cells is $2^{k+d}2^{k+d}=2^{2k}2^{2d}$.
As $2^{2d}$ grows faster that $6\cdot2^d$, for large enough $d$ (actually for $d=3$) 
the total number of cells is larger than the number of disabled or non-free cells.  

The case $c=0$ is done.
For arbitrary $c$ we 
change the construction a little bit by  
letting $U_c(1^{c+d+2}q)=U(q)$ instead of  
$U_0(1^{d+2}q)=U(q)$. The optimal machine $U_c$ constructed in this way depends on 
$c$, which is inevitable by Theorem~\ref{t:complist}.

\section{Other applications of explicit graphs with on-line matching}
\label{s:other}

Our applications are related to the resource bounded Kolmogorov complexity.
Recall that a machine $U$ is called standard if for any machine $V$ there is a 
total computable function $f$ such that $U(t(p),z)=V(p,z)$ and $|t(p)|\le |p|+\bigO{1}$ for all $p,z$.
In this section we assume that  $t$ is polynomial-time computable and that 
running time of  
$U(t(p),z)$ is bounded by a polynomial of the computation time of $V(p,z)$. 
By $C_U^T(x|z)$ we denote the minimal length of $p$ such that 
$U(p,z)=x$ in at most $T$ steps.

The applications consists in improving Muchnik's Theorem,  and the error term in the estimation of the distinguishing complexity of strings in a given set.

\subsection{Muchnik's theorem}

\begin{theorem}[Muchnik's Theorem~\cite{muc:j:condcomp,mus-rom-she:j:muchnik}]
\label{t:muchnik}
Let $a$ and $b$ be strings such that $|a|=n$ and $C(a \mid b) = k$. 
Then there exists a string $p$ such that (1) $|p| = k + \bigO{\log n}$, (2) $C(p \mid a) = \bigO{\log n}$, (3) $C(a \mid p,b) = \bigO{\log n}$. 
\end{theorem}
 See the cited works for a discussion of Muchnik's theorem.
\begin{theorem}[Improved version of Theorem~\ref{t:muchnik}]
Same statement as above except that we replace (2) by (2') $C^{q(n)}( p \mid a) = \bigO{\log n}$, where $q$ is a polynomial.
\end{theorem}
\begin{proof}
Fix an explicit graph $G = (L,R,E)$,  with $L=\zo^{n}$, polynomial left degree, and that has computable on-line matching
with logarithmic overhead (such a graph is obtained in the proof of Theorem~\ref{th2}).  Given a string $b$, 
run the  optimal machine $U(q,b)$ in parallel for all $q$. Once, for some $q$,  $U(q,b)$ halts with
the result $x$,  pass the request $(x,|q|)$ to the matching algorithm in the graph.
It will return a neighbor $p$ of length at most $|q|+\bigO{\log n}$ of $x$.  
At some moment a shortest program $q$ for $a$ conditional to $b$ will halt 
and we get the sought $p$.

As the graph is explicit and has polynomial degree, we 
have  $C^{\poly(n)}(p\mid a)=\bigO{\log n}$ (requirement (2')).
Requirement (1) holds by construction. Finally, $C(a \mid p,b) = \bigO{\log n}$
as given $p$ and $b$ we may identify $a$ by running the above algorithmic process (it is important 
that $a$ is the unique string that was matched to $p$).
\end{proof}

\medskip

\subsection{Distinguishing complexity.}

The $T$-bounded distinguishing complexity of a string $x$, introduced by Sipser~\cite{sip:c:randomness},  is the length of a minimal program that in $T$ steps  accepts $x$ and only $x$.
Formally, let $U$ be a machine, $x$ a string and $T$ a natural number.
The \emph{distinguishing complexity $CD^T_U(x)$ with respect to $U$} is defined 
as the minimal length of $p$ such that $U(p,x)=1$ ($p$ ``accepts'' $x$) in at most $T$ steps, and $U(p,x')=0$ 
for all $x'\ne x$ ($p$ ''rejects'' all other strings after any number of steps). 
From our assumption for the standard machine $U$,  it follows that for every machine $V$ there is a polynomial $f$ and a constant $c$ such 
that $CD^{f(T)}_U(x)\le CD^T_V(x)+c$. Indeed, let $p$ be a shortest distinguishing program for $x$ working in $T$ steps
with respect to $V$. Then $t(p)$ is a program for $U$ that accepts $x$ in $\poly(T)$ steps and rejects all other strings.

For a set $A$ of binary strings,  let $A^{=n}$ stand for the set of all strings of length $n$ in $A$. Buhrman, Fortnow and Laplante~\cite{bfl:j:boundedkolmogorov} have shown that most
strings have polynomial-time distinguishing complexity close to the information-theoretic minimum value.

\begin{theorem}[\cite{bfl:j:boundedkolmogorov}]
\label{t:disting}
For every function $\epsilon(n)$ (mapping natural numbers 
to numbers of the form $1/$natural) computable in time  $\poly(n)$ there is a polynomial $f$ such that 
for every set $A$, for all $x \in A^{=n}$ except for a fraction $\epsilon(n)$,
$CD^{f(n),A}_U(x) \leq \log|A^{=n}| + {\rm polylog}(n/\eps(n))$.
\end{theorem}

We mean here that the set $A$ is given to the standard machine $U$ as an oracle (so we assume that the standard machine is an oracle machine
and all the requirements hold for every oracle.)

Our improvement of the above theorem does not use explicitly graphs with on-line matching. However, it uses an argument similar to the one in
Theorem~\ref{th2} to obtain a graph with a certain ``low-congestion'' property, which would allow most nodes making matching requests to have their 
``reserved'' matching node.
\begin{theorem}[Improved version of Theorem~\ref{t:disting}]
Same statement as above, except that  we obtain $CD^{f(n),A}(x) \leq \log|A^{=n}| 
+ \bigO{\log(n/\eps(n))}$, \ie, we reduce the error term from $ {\rm polylog}(n/\eps(n))$ to $ \bigO{\log(n/\eps(n))}$ .
\end{theorem}
\begin{proof}
For our improvement we need for every $n$, $k\le n$ and $\eps$ 
a bipartite graph  $G_{n,k,\eps}$ with $L=\zo^n$, $R=\zo^{k+\bigO{\log n/\eps}}$ and degree $\poly(n/\eps)$ that has the following ``low-congestion"  property:
\begin{quote}
for every subset $S$ of at most $2^k$ left nodes   
for every node $x$ in $S$ except for a fraction $\eps$ there is a right neighbor $p$ of $x$ 
such that $p$ has no other neighbors in $S$. 
\end{quote}

Assume that we have such an explicit family of graphs $G_{n,k,\eps}$. 
Here, explicit means that given $n$, $k$, $\eps$,  a left node $x$ and $i$,  we can in polynomial time find 
the $i$th neighbor of $x$ in  $G_{n,k,\eps}$.
Then we can construct 
a machine $V$ that,  given a tuple $(p,i,n,k)$, a string $x$ and $A$ as oracle,  
verifies that $x$ is in $A^{=n}$ and that $p$ is the $i$th neighbor of $x$ in $G_{|x|,k,\eps(n)}$. 
If this is the case it accepts and rejects otherwise.
By the property of the graph, applied to $S=A^{=n}$ and $k=\lceil\log|S|\rceil$ we see that
$$
CD^{f(n),A}_V(x)\le|(p,i,n,k)|\le \log |A^{=n}|+\bigO{\log n/\eps(n)}
$$
for some polynomial $f(n)$ for all but a fraction $\eps(n)$ for $x\in A^{=n}$.
By the assumptions on $U$,  the same inequality holds for $U$.

The graph $G_{n,k,\eps}$ is again obtained from the disperser of Ta-Shma et al.~\cite{ats-uma-zuc:j:expanders}.
Given $n$, $k$ and $\eps$,  we  apply Theorem~\ref{t:tuz} to $K=\eps 2^k$ and $\delta=1/2$.
We obtain a \ver{$(\eps 2^k, \frac{\alpha \eps 2^k D}{2n^3})$}-expander with degree $D=\poly(n)$, $L=\zo^n$ and
\ver{$|R| = \frac{\eps\alpha 2^k D}{n^3}$}.
Consider $t=\max\{1,\lceil \frac{2n^3}{\alpha D}\rceil\}$ disjoint copies of this graph and identify left
nodes of the resulting graphs (keeping their sets of right nodes disjoint). 
We get an explicit  $(2^k\eps, 2^k\eps)$-expander with $2^n$ left and \ver{$2^{k}\poly(n)\eps$} right nodes
and degree \ver{$D=\poly(n)t=\poly(n)$.} 

This graph, called $H_{n,k,\eps}$, has the following weaker ``low-congestion"  property:
\emph{for every set of $2^k$ left nodes $S$ for every node $x$ in $S$ except for a fraction $\eps$ there is a right neighbor $p$ of $x$ 
such that $p$ has at most $D/\eps$ neighbors in $S$.} 

Indeed, the total number of edges in the graph originating in  $S$ is at most $|S|D$. 
Thus less than $|S|D/(D/\eps)=|S|\eps$ right nodes are ``fat'' in the sense that they have more 
than  $D/\eps$ neighbors landing in $S$. By the expander property of $H_{n,k,\eps}$ there are less than $\eps|S|$
left nodes in $S$ that have only fat  neighbors.

It remains to ``split'' right nodes of $H_{n,k,\eps}$ so that $D/\eps$ becomes 1.
This is done exactly as in Buhrman et al.~\cite{bfl:j:boundedkolmogorov}. Using the Prime Number Theorem, it is not hard  
to show (Lemma 3 in Buhrman et al.~\cite{bfl:j:boundedkolmogorov}) that for every set $W$ of $d$ strings of length $n$
the following holds:
\emph{for every $x\in W$ there is a prime number $q\le 4dn^2$ such that $x\not\equiv x' \pmod q$   
for all $x'\in W$ different from $x$} (we identify here natural numbers and their binary expansions).

We apply this lemma to $d=D/\eps$. 
To every right node $p$ in $H_{n,k,\eps}$ we add a prefix code of two natural numbers $a,q$, both at  most 
$4dn^2$, and connect a left node $x$ to $(p,a,q)$ if $x$ is connected to $p$ in $H_{n,k,\eps}$
and $x\equiv a \pmod q$. We obtain the graph $G_{n,k,\eps}$ we were looking for.
Indeed, for every $S$ of $2^k$ left nodes for all $x\in S$ but a fraction 
of $\eps$ there is a neighbor $p$ of $x$ in $H_{n,k,\eps}$ that has at most $d=D/\eps=\poly(n)/\eps$
neighbors in $S$. Besides there is a prime $q\le 4n^2d= \poly(n)/\eps$ such that 
$x\not\equiv x' \pmod q$   
for all neighbors $x'$ of $p$ different from $x$. Thus the neighbor $(p,q,x \bmod q)$ of $x$ in   $G_{n,k,\eps}$
has no other neighbors in $S$.

The degree of $G_{n,k,\eps}$ is $D\times (4n^2D/\eps)^2=\poly(n)/\eps^2$. The number of right
nodes is 
\ver{$$
(\poly(n)2^k\eps)(4n^2D/\eps)^2=2^k\poly(n)/\eps.
$$} 
Thus right nodes can be identified with strings of length $k+\bigO{\log n/\eps}$ and we are done.
\end{proof}

\medskip

\if01
\textbf{Compression of samplable sources, Theorem 6.1~\cite{tre-vad-zuc:j:compression}.} Let $X_n$ be a flat source with membership algorithm and $H(X_n) \leq k$. Then $X_n$ is compressible to length $k + {\rm polylog}(n-k)$ in time $\poly(n, 2^{n-k})$.

In our version, the compression length becomes $k + \bigO{\log n}$.

\marginpar{We need a rigorous formulation and a proof}
\fi

\section{Acknowledgments} We are grateful to Andrei Romashchenko and Alexander (Sasha) Shen 
for useful discussion. 
We thank Jason Teutsch for stimulating the investigation of
Theorem~\ref{t:complist} with higher precision than $\bigO{\log |x|}$.
We also thank the anonymous reviewers for many detailed comments on the manuscript.


\bibliographystyle{plain}
\bibliography{theory}

\end{document}